\documentclass[12pt,a4paper,reqno,oneside]{amsart}
\usepackage{t1enc}
\usepackage{amssymb}
\usepackage[mathscr]{euscript}
\usepackage{color}
\usepackage{float}
\usepackage{graphicx}
\usepackage{url}
\usepackage{caption}
\usepackage{subcaption}
\usepackage{comment}
\usepackage{bm}

\numberwithin{equation}{section}

\newcommand{\cev}[1]{\reflectbox{\ensuremath{\vec{\reflectbox{\ensuremath{#1}}}}}}

\newtheorem{defi}{Definition}[section]
\newtheorem{thm}[defi]{Theorem}
\newtheorem{lemm}[defi]{Lemma}
\newtheorem{rem}[defi]{Remark}

\newtheorem{exam}[defi]{Example}
\newtheorem{nota}[defi]{Notation}
\newtheorem{ass}[defi]{Assumption}

\newcommand{\RR}{\mathbb{R}}
\newcommand{\EE}{\mathbb{E}}

\newcommand{\PP}{\mathbb{P}}
\newcommand{\QQ}{\mathbb{Q}}

\newcommand{\R}{\mathbb{R}}

\begin{document}

\title[A functional variational approach to pricing path dependent policies]
{A functional variational approach to pricing path dependent insurance policies}
\author[D. Ba\~{n}os]{David Ba\~{n}os}
\address{D. Ba\~{n}os: Department of Mathematics, University of Oslo, Moltke Moes vei 35, P.O. Box 1053 Blindern, 0316 Oslo, Norway.}
\email{davidru@math.uio.no}
\author[S. Ortiz-Latorre]{Salvador Ortiz-Latorre}
\address{S. Ortiz-Latorre: Department of Mathematics, University of Oslo, Moltke Moes vei 35, P.O. Box 1053 Blindern, 0316 Oslo, Norway.}
\email{salvadoo@math.uio.no}
\author[O. Zamora]{Oriol Zamora Font}
\address{O. Zamora: Department of Mathematics, University of Oslo, Moltke Moes vei 35, P.O. Box 1053 Blindern, 0316 Oslo, Norway.}
\email{oriolz@math.uio.no}


\begin{abstract}
The main purpose of this work is the derivation of a functional partial differential equation (FPDE) for the calculations of equity-linked insurance policies, where the payment stream may depend on the whole past history of the financial asset. To this end, we employ variational techniques from the theory of functional Itô calculus. 
\end{abstract}

\maketitle



\section{Introduction}

In actuarial science, a reserve is a liability equal to the present value of future cash flows that the insurance company promises to pay out to the insured under certain conditions. In easier terms, a reserve is an estimate of \emph{how much} the insurance company should charge today in order to meet future payments owed to the insured. This quantity is the basis to ensure the solvency of the company and the standard way of computing premiums. The literature sometimes distinguishes between the terms \emph{net present value} of future obligations and \emph{actuarial reserve}. The former is the \emph{expected} cost of the insurance itself, while the latter is obtained by subtracting the premiums provided by the insured, which are used to pay back to the insured. In this sense, the paid-in premiums should match, in average, the future obligations in such a way that the reserve at the entry of the contract is null. We point out that practitioners in the industry often use the terminology \emph{reserve} as the upper quantile of the loss distribution for each exercise year, usually $99.5\%$ as stipulated in the Solvency II delegated acts. In this work, however, we focus on the academic definition. 

In life insurance, the main sources of risk of a policy are the state of the insured which triggers the payments, and the future development of the return on the financial investments or technical interest rate. In this note, we focus on the latter and model its risk via a continuous time Itô-diffusion of non-Markovian type. We suggest a way to make cash flows equity dependent and provide some new examples where standard calculus cannot be applied due to the lack of Markovianity in the financial market. This process can be thought of as an investment fund where premiums are deposited, or simply the evolution of some stock or pool of stocks that the insured wishes to invest in. The techniques and results of our work could be adapted for stochastic technical interest rate as well, including stochastic interest rates that are non-Markovian. We choose to present and develop a modelling framework for equity-linked policies because it allows for more generality, since payments can be dependent on market movements.

Life insurance claims in the context of stochastic interest rates have been studied before. We mention \cite{Bacinello93, Bacinello94} for some specific interest rate models and \cite{Bacinello, Kurtz96, Nielsen95} for more general models in the framework of Heath-Jarrow-Morton. A model for the interest rate of diffusion type was considered by Norberg and Møller in \cite{Norberg96} and by the authors in \cite{Banos20} where they look at unit-linked insurances with variance risk, as well.

On the other hand, research on unit-linked policies, that is policies where payments are stochastic due to the influence of market investments, has been widely studied in the Markov setting. For instance, the already mentioned works \cite{Bacinello93, Bacinello94} consider also equity-linked contracts. Their modelling framework for interest rates is the well-known Heath-Jarrow-Morton which is more general than just Markovian diffusions. However, the stock evolution is the classical Black-Scholes model with two independent noises.

One of the seminal works on this topic is due to Brennan and Schwartz in \cite{Brennan76}. Also in \cite{Aase94} Aase and Persson derive a partial differential equation for computing reserves of equity-linked policies in a risk-neutral framework, still on a Markovian setting. The latter work by Aase and Persson is developed further towards a more general framework by Steffensen in \cite{Steffensen00}, where the author uses classical principles of no-arbitrage imported from financial mathematics to price insurance contracts in the context of securitization. In this way, the author finds and provides the right connection between insurance pricing and financial mathematics. If one assumes that insurance contracts can be traded in the market by securitization, then the no-arbitrage approach provides the right way of pricing such instruments and a Thiele's equation is derived. Actually, since the principles from \cite{Steffensen00} are rather model-free and more of philosophical nature, they can also be used, and will be used, tacitly, in our framework as well. Although path dependent payoffs are discussed in \cite{Steffensen00} in form of Asian-type options, the underlying assets are assumed to be Markovian from the very beginning. In this sense, the present work generalizes this property to path dependent processes.

A classical way of computing reserves in the continuous time setting is by solving the so-called Thiele's differential equation. In the case of deterministic interest rates and policy payments, it is an ordinary differential equation, while in the case of equity-linked insurance policies, it is a partial differential equation (PDE). The corresponding Thiele's equation for the case of stochastic rates was derived by Norberg and Møller in \cite{Norberg96} and later risk adjusted by Persson in \cite{Persson98}.


However, the assumption of non-Markovianity in our underlying stochastic processes calls for more advanced techniques in order to treat functional SDEs. We show that an analogue of Thiele's equation resembling that one in \cite{Steffensen00} can be retrieved by resorting to functional Itô calculus developed by Dupire \cite{Dupire09} and further developed by Cont \cite{Cont16}. The reader may consult \cite{Cont16} for a course on this topic. Using the variational techniques from \cite{Dupire09,Cont16} one can show that insurance contracts that depend on the whole performance of a financial asset can be priced and hedged. The formula we obtain is an expected infinite-dimensional PDE, that is, a PDE where the space variable belongs to some metric space of paths. The derivatives involved in the equation \emph{are not} Fréchet derivatives, but rather some type of directional derivatives that make computations more applicable.
Our framework includes, not only all classical types of insurance contracts such as Guaranteed Minimum Maturity Benefits (GMMB) but also path dependent options such as Asian or lookback options or even more complex path dependent options. Thus, the derived formula provides an alternative method to price and compute reserves for such type of insurance contracts.


The paper is organized as follows: In Section \ref{framework} we introduce our modelling framework. This is a friendly introduction to the functional Itô calculus and the probability spaces carrying the financial and insurance information. In Section \ref{sec:thiele} we present the hypothesis in our model and derive the functional Thiele's equation.

\section{Mathematical framework}\label{framework}

In this section we present the mathematical context in which the main results of the paper lie. The framework can be divided into two main parts: mathematical finance and insurance, and a functional variational calculus due to Dupire \cite{Dupire09}. For the insurance setting, we will mainly adopt the modelling framework of Norberg, see \cite{Norberg91}, in connection with equity-linked insurance products in the same flavour as in \cite{Aase94,Koller12,Kurtz96}.

We will hereunder offer an introduction on the functional Itô calculus and later on we will present our financial and actuarial model.

\subsection{Functional Itô calculus} 

We start by giving a brief résumé of the theory of functional Itô calculus due to Dupire \cite{Dupire09}. The reader may consult \cite{Cont16} for a more detailed introduction to the topic.

The functional Itô calculus is a non-anticipative (as opposed to e.g. Malliavin calculus) stochastic calculus of variations for functionals of the whole path of a (not necessarily continuous) semimartingale. This calculus was introduced by Bruno Dupire in \cite{Dupire09} motivated by financial applications and it overcomes the limitations of the usual Itô calculus which cannot, in general, deal with path dependent options. For example, this calculus is very useful for path dependent options such as Asian options or basket options.

Even though the theory is extensive, we will simply introduce the required mathematical tools for our purposes. Essentially, we will present the necessary definitions and results from \cite{Cont16}, following the notation accordingly.

\begin{nota}
   Let $T>0$ be a fixed time horizon and $D([0,T],\RR)$ be the space of càdlàg functions defined on $[0,T]$ with values in $\RR$. For a càdlàg function $\omega\in D([0,T],\RR)$, we denote by
    \begin{enumerate}
        \item $\omega(t)\in\RR$, the value of $\omega$ at $t\in[0,T]$;
        \item $\omega_t\in D([0,T],\RR)$ the path stopped at $t\in[0,T]$, that is, $\omega_t(s)=\omega(t\wedge s)$ for $s\in[0,T]$;
        \item $\omega(t-)=\lim_{\substack{s\to t \\ s<t}}\omega(s)$, the left limit at $t\in[0,T]$.
    \end{enumerate}
   For a càdlàg stochastic process $X=\{X(t), t\in[0,T]\}$ we similarly denote by
    \begin{enumerate}
        \item $X(t)$, its value at $t\in[0,T]$;
        \item $X_t=\{X(t\wedge s), s\in[0,T]\}$, the process stopped at $t\in[0,T]$;
        \item $X(t-)=\lim_{\substack{s\to t \\ s<t}}X(s)$, the left limit at $t\in[0,T]$. 
    \end{enumerate}
    In addition, we will denote by $C^0([0,T],\RR)$ the space of continuous functions defined on $[0,T]$ with values in $\RR$.
\end{nota}

Let $\Omega=D([0,T],\RR)$, $X$ the canonical process on $\Omega$, that is, $X(t,\omega)=\omega(t)$ for $\omega\in\Omega$, and $\mathcal{F}^X=\{\mathcal{F}^X_t\}_{t\in[0,T]}$ the filtration generated by $X$. 

A process $Y$ adapted to $\mathcal{F}^X$ can be represented by a functional $F$ of the stopped path of $X$ at $t$. Namely,
$$Y(t,\omega)= F(t,X_t)=F(t,\omega_t),$$
where $F:[0,T]\times \Omega\rightarrow \R$. Note that $F(t,\cdot)$ actually only needs to be defined on the set of stopped paths at $t$. A stopped path is an equivalent class in $[0,T]\times \Omega$ with the equivalence relation
$$(t,\omega)\sim (t',\omega') \iff t=t' \mbox{ and } \omega_t=\omega'_{t}.$$
The space of stopped paths can be defined as $\Lambda_T \triangleq [0,T]\times \Omega/\sim$ which we equip with the metric given by
\begin{align*}d_{\infty}((t,\omega),(t',\omega')) &=\sup_{s\in [0,T]}|\omega(t\wedge s)-\omega'(t'\wedge s)|+|t-t'|.\end{align*}
The metric space $(\Lambda_T,d_{\infty})$ is then a complete metric space and the subspace of continuous stopped paths $\mathcal{W}_T=\{(t,\omega)\in\Lambda_T, \omega\in C^0([0,T],\RR)\}$ is a closed subset of it. In summary, we can view adapted processes as functionals on the space of stopped paths $\Lambda_T$, and vice versa.

We now give the mathematical definition of a non-anticipative functional and the notion of continuous and left-continuous non-anticipative functionals. 

\begin{defi}[Non-anticipative functional]
A non-anticipative functional on $\Omega$ is a Borel measurable map $F: (\Lambda_T ,d_\infty)\rightarrow \R$ on the space $(\Lambda_T ,d_\infty)$ of stopped paths.
\end{defi}

Note that for $\omega\in D([0,T],\RR)$ and $F$ a non-anticipative functional, we have that $F(t,\omega)=F(t,\omega_t)$ since $(t,\omega)\sim (t,\omega_t)$. 

\begin{defi}[$\mathbb{C}^{0,0}(\Lambda_T)$ functionals] A continuous non-anticipative functional is a continuous map $F: (\Lambda_T, d_\infty)\rightarrow \RR$. That is, $\forall (t,\omega_t)\in\Lambda_T$, $\forall \epsilon>0$, $\exists\eta>0$, $\forall(t',\omega'_{t'})\in\Lambda_T$
\begin{align*}
    d_\infty((t,\omega),(t',\omega'))<\eta \implies |F(t,\omega_t)-F(t',\omega'_{t'})|<\epsilon.
\end{align*}
    We denote by $\mathbb{C}^{0,0}(\Lambda_T)$ the set of continuous non-anticipative functionals. 
\end{defi}

\begin{defi}[$\mathbb{C}_l^{0,0}(\Lambda_T)$ functionals] A left-continuous non-anticipative functional is a non-anticipative functional such that: $\forall (t,\omega_t)\in\Lambda_T$, $\forall\epsilon>0$, $\exists\eta>0$, $\forall(t',\omega'_{t'})\in\Lambda_T$,
\begin{align*}
    t'<t \text{  and  } d_\infty((t,\omega),(t',\omega'))<\eta \implies |F(t,\omega_t)-F(t',\omega'_{t'})|<\epsilon.
\end{align*}
     We denote by $\mathbb{C}_l^{0,0}(\Lambda_T)$ the set of left-continuous non-anticipative functionals.    
\end{defi}

We introduce the derivatives of the functional Itô calculus, namely, the horizontal and vertical derivative. The horizontal derivative is, informally, the variation of a non-anticipative functional along a horizontal extension of the stopped path $(t,\omega_t)$ to $(t+h,\omega_t)$. On the other hand, the vertical derivative can be seen as a particular case of a Gâteaux derivative along a vertical perturbation from $(t,\omega_t)$ to $(t,\omega_t+h\mathbb{I}_{[t,T]})$.

\begin{defi}[Horizontal and vertical derivatives]
A non-anticipative functional $F : \Lambda_T \rightarrow \R$ is said to be horizontally differentiable at $(t,\omega_t)\in \Lambda_T$ if the limit
$$\mathcal{D}F(t,\omega_t)= \lim_{\substack{h\to 0 \\ h>0}} \frac{F(t+h,\omega_t)-F(t,\omega_t)}{h}$$
exists and is finite. If $F$ is horizontally differentiable at all $(t,\omega_t)\in\Lambda_T$, then $\mathcal{D}F$ is a non-anticipative functional called the horizontal derivative of $F$. 

Further, $F$ is said to be vertically differentiable at $(t,\omega_t)\in \Lambda_T$ if the limit
$$\nabla_\omega F(t,\omega_t)= \lim_{h\to0} \frac{F(t,\omega_t+h\mathbb{I}_{[t,T]})-F(t,\omega_t)}{h}$$
exists and is finite. If $F$ is vertically differentiable at all $(t,\omega_t)\in \Lambda_T$ then $\nabla_\omega F$ is a non-anticipative functional called the vertical derivative of $F$.
\end{defi}

Note that $\nabla_\omega F(t,\omega_t)$ is the directional derivative of $F(t,\cdot)$ in the direction $\mathbb{I}_{[t,T]}$. It is worth mentioning that the notion of vertical derivative is much weaker than the Fréchet derivative which involves perturbations in all directions and not only in the direction $\mathbb{I}_{[t,T]}$.

Observe that if $F(t,\omega_t)=f(t,\omega(t))$ with $f\in C^{1,1}([0,T]\times\R,\R)$ then $\mathcal{D}F(t,\omega_t)= \partial_t f(t,\omega(t))$ and $\nabla_\omega F(t,\omega_t)=\partial_x f(t,\omega(t))$. Therefore, we can view the horizontal and vertical derivative as extensions of the notions of partial derivative in time and space, respectively.

Previous to presenting the functional Itô formula we need to introduce the following space of regular functionals.

\begin{defi}[$\mathbb{B}(\Lambda_T)$ functionals]
Define $\mathbb{B}(\Lambda_T)$ as the set of non-anticipative functionals $F:\Lambda_T \rightarrow \R$ such that for any compact set $K\subset \R$ and $t_0<T$, there exists a finite constant $C_{K,t_0}>0$ such that for all $t\leq t_0$, $\omega\in \Omega$ and $\omega([0,t])\subset K$ implies $F(t,\omega_t)\leq C_{K,t_0}$.
\end{defi}

\begin{defi}[$\mathbb{C}_b^{1,2}(\Lambda_T)$ functionals]
Define $\mathbb{C}_b^{1,2}(\Lambda_T)$ as the set of left-continuous functionals $F\in \mathbb{C}_l^{0,0}(\Lambda_T)$ such that
\begin{itemize}
\item $F$ admits a horizontal derivative $\mathcal{D}F(t,\omega_t)$ for all $(t,\omega_t)\in \Lambda_T$ and $\mathcal{D}F(t,\cdot): D([0,T],\RR)\rightarrow \R$ is continuous for each $t\in [0,T]$ under the uniform topology.
\item $\nabla_\omega F,\nabla_\omega^2 F\in \mathbb{C}_l^{0,0}(\Lambda_T)$.
\item $\mathcal{D}F,\nabla_\omega F,\nabla_\omega^2 F\in \mathbb{B}(\Lambda_T)$.
\end{itemize}
\end{defi}

\begin{defi}[$\mathbb{C}_{\text{loc}}^{1,2}(\Lambda_T)$ functionals] A non-anticipative functional $F\in\mathbb{C}^{0,0}_b(\Lambda_T)$ is said to be locally regular if there exists an increasing sequence of stopping times $\tau_0=0$, $\tau_k\uparrow\infty$, and $F^k\in\mathbb{C}_b^{1,2}(\Lambda_T)$ such that
\begin{align*}
    F(t,\omega_t)=\sum_{k\geq0}F^k(t,\omega_t)\mathbb{I}_{[\tau_k(\omega),\tau_{k+1}(\omega))}(t).
\end{align*}
We denote by $\mathbb{C}_{\text{loc}}^{1,2}(\Lambda_T)$ the set of locally regular non-anticipative functionals.
\end{defi}

Let $Y$ be an $\mathcal{F}^X$-adapted process given by a non-anticipative functional $Y(t)=F(t,X_t)$ where $F\in\mathbb{C}_b^{1,2}(\Lambda_T)$. Assume, as it will be of interest later, that $X$ has continuous paths. Then, $Y$ only depends on the restriction of $F$ to 
\begin{align*}
    \mathcal{W}_T= \{(t,\omega)\in\Lambda_T, \omega\in C^0([0,T],\RR)\}. 
\end{align*}
So, the representation $Y(t)=F(t,X_t)$ may not be unique. Moreover, the computation of the vertical derivative involves the evaluation of $F$ outside $\mathcal{W}_T$ due to adding the jump $h\mathbb{I}_{[t,T]}$ to the path $\omega_t$. However, \cite{Cont16} shows that this does not create any inconsistency in the computation of the vertical derivative. More precisely, Theorems 5.4.1 and 5.4.2 in \cite{Cont16} show that if two functionals agree on continuous paths, then their vertical derivatives are the same. Hence, the choice of $F$ in the representation $Y(t)=F(t,X_t)$ is consistent and we can define the following class of regular functionals. 

\begin{defi}[$\mathbb{C}^{1,2}_b(\mathcal{W}_T)$ and $\mathbb{C}^{1,2}_{loc}(\mathcal{W}_T)$ functionals] 
We say that $F\in\mathbb{C}^{1,2}_b(\mathcal{W}_T)$ if there exists $\widetilde{F}\in\mathbb{C}_b^{1,2}(\Lambda_T)$ such that $\widetilde{F}_{|\mathcal{W}_T}=F$. We define $\mathbb{C}^{1,2}_{loc}(\mathcal{W}_T)$ in a similar way.     
\end{defi}

We introduce the functional Itô formula, which resembles the standard Itô formula but replacing temporal and spatial derivatives with horizontal and vertical derivatives, respectively. Note however that the following Itô formula holds pathwise, in contrast to the classical one.

\begin{thm}[Functional Itô formula for continuous semimartingales]
Assume that $X$ is a continuous semimartingale and $F\in \mathbb{C}_{loc}^{1,2}(\mathcal{W}_T)$. For any $t\in [0,T)$,
\begin{align*}
F(t,X_t) =&\, F(0,X_0)+ \int_0^t \mathcal{D} F(u,X_u)du\\
&+ \int_0^t \nabla_\omega F(u,X_u) dX(u) + \frac{1}{2}\int_0^t \nabla_\omega^2 F(u,X_u)d[X](u),\quad \PP-\mbox{a.s.}
\end{align*}
where $[X]$ stands for the quadratic variation of the (continuous) semimartingale $X$.
\end{thm}

We give now sufficient conditions for existence and uniqueness of strong solutions of path dependent SDEs. Note the similarities with the Lipschitz and linear growth conditions for the existence and uniqueness of strong solutions of standard SDEs.

\begin{thm}[Existence and uniqueness of path dependent SDEs]\label{thm:E&U}
Let $t_0\in[0,T]$, $b,\sigma\in\mathbb{C}_l^{0,0}(\Lambda_T)$ satisfying the following Lipschitz and linear growth conditions
\begin{align}\label{eq:pdpde1}|b(t,\omega_t)-b(t,\omega'_{t})|+|\sigma(t,\omega_t)-\sigma(t,\omega'_{t})|\leq K \sup_{s\leq t}|\omega(s)-\omega'(s)|,\end{align}
and
\begin{align}\label{eq:pdpde2}|b(t,\omega_t)|+|\sigma(t,\omega_t)|\leq K\left(1+ \sup_{s\leq t}|\omega(s)|+|t|\right),\end{align}
for all  $\omega,\omega'\in C^0([0,t],\R)$, all $t\geq t_0$ and let $W$ be a standard Brownian motion. Then, for any $\xi\in C^0([0,T],\R)$ the path dependent SDE
\begin{align}\label{eq:pdsde1}
    dX(t)=b(t,X_t)dt+\sigma(t,X_t)dW(t), \quad X_{t_0}=\xi_{t_0},
\end{align}
has a unique strong solution. The paths of $X$ lie in $C^0([0,T],\R)$ and there exists a constant $C$ depending only on $T$ and $K$ such that, for $t\in[t_0,T]$,
\begin{align*}
    \EE\left[\sup_{s\in[0,t]}X(s)^2\right]\leq C\left(1+\sup_{s\in[0,t_0]}\xi(s)^2\right)e^{C(t-t_0)}.
\end{align*}
\end{thm}

In the connection between a diffusion and its Kolmogorov equation, the domain of the PDE is related to the support of a random variable. For the analogous case in functional Itô calculus, we need to introduce the concept of topological support of a (continuous) stochastic process. 

\begin{defi}[Topological support of a stochastic process] Let $Y=\{Y(t), t\in[0,T]\}$ be a continuous stochastic process. We define the topological support of $Y$ in $(C^0([0,T],\RR),\|.\|_\infty)$ as the following set of continuous paths
\begin{align*}
    supp(Y)=\{\omega\in C^0([0,T],\RR) \ | \ \forall \text{ Borel neighborhood } V \text{ of } \omega, \ \PP(Y_T\in V)>0\}.
\end{align*}
\end{defi}

Finally, we introduce a result that will be useful for deriving a path dependent Thiele's equation. Essentially, it establishes a connection between local martingales and solutions of certain path dependent PDEs. 

\begin{lemm}\label{lem:locmar} Let $b,\sigma\in\mathbb{C}_l^{0,0}(\Lambda_T)$ satisfying the conditions in \eqref{eq:pdpde1} and \eqref{eq:pdpde2} and let $X$ be the unique strong solution of the path dependent SDE \eqref{eq:pdsde1}. If $F\in\mathbb{C}_b^{1,2}(\mathcal{W}_T)$ and $\mathcal{D}F\in\mathbb{C}_l^{0,0}(\mathcal{W}_T)$, then $Y(t)=F(t,X_t)$ is a local martingale if and only if $F$ satisfies
\begin{align*}
    \mathcal{D}F(t,\omega_t)+b(t,\omega_t)\nabla_\omega F(t,\omega_t)+\frac{1}{2}\sigma(t,\omega_t)^2\nabla_\omega^2F(t,\omega_t)=0
\end{align*}
on the topological support of $X$ in $(C^0([0,T],\RR), \|.\|_\infty)$.
\end{lemm}

After introducing the necessary results on functional Itô calculus, we now define the framework that will be used throughout the paper. We start introducing the financial model and then the insurance one.

\subsection{The financial model}\label{sec:fin}
Let $T\in \mathbb{R}$, $T>0$ be a time horizon, typically the maturity of a contract. On a complete probability space $(\Omega,\mathcal{A},\PP)$ we consider a standard Brownian motion $W$ and non-anticipative functionals $\widetilde{b},\widetilde{\sigma}$ such that the functionals $b(t,\omega_t)=\omega(t)\widetilde{b}(t,\omega_t)$ and $\sigma(t,\omega_t)=\omega(t)\widetilde{\sigma}(t,\omega_t)$ satisfy $b,\sigma\in\mathbb{C}_l^{0,0}(\Lambda_T)$ and the conditions given in \eqref{eq:pdpde1} and \eqref{eq:pdpde2} of Theorem \ref{thm:E&U}. Then, we consider $S=\{S(t),t\in [0,T]\}$ the unique strong solution of the following \emph{path dependent} SDE,
\begin{align*}
\frac{dS(t)}{S(t)} = \widetilde{b}(t,S_t)dt + \widetilde{\sigma}(t,S_t)dW(t),\quad t\in [0,T],\quad S(0)=x\in \R,
\end{align*}
as the value of a risky asset or fund. Note that the previous path dependent SDE can be simply written as
\begin{align*}
    dS(t)=b(t,S_t)dt+\sigma(t,S_t)dW(t).
\end{align*}

\begin{ass}
    We assume that $S$ and $W$ generate the same filtrations. Note that the family of $\widetilde{b}$ and $\widetilde{\sigma}$ that satisfy this property is quite rich. One simple example is a strictly positive constant volatility and a function of the integral of $S$ as a drift.
\end{ass}

In addition, we consider a technical discount factor $v=\{v(t),t\in [0,T]\}$ modelling the company's current value of one monetary unit to be exercised at time $t$, given by $v(t)=e^{-\int_0^t r(s)ds}$ for a technical rate of return $r$. We hereby assume that $r$ is non-negative, deterministic and càdlàg. The assumption of a deterministic $v$ could be relaxed, but the mathematics become more involved.

\begin{ass}
    We assume that this model is arbitrage-free and complete. Therefore, there exists a unique equivalent martingale measure $\QQ$. For instance, one sufficient condition could be that the market price of risk $\theta(t)=\frac{\widetilde{b}(t,S_t)-r(t)}{\widetilde{\sigma}(t,S_t)}$ satisfies the Novikov condition. 
\end{ass}
Note that path dependent SDE of $S$ under $\QQ$ is given by
\begin{align*}
    \frac{dS(t)}{S(t)}=r(t)dt+\widetilde{\sigma}(t,S_t)dW^\QQ(t),
\end{align*}
where $W^\QQ$ is a $\QQ$-Brownian motion. 

\subsection{The insurance model}
On the same complete probability space $(\Omega,\mathcal{A},\PP)$ introduced in Section \ref{sec:fin}, we describe the states of a policyholder by a càdlàg pure jump process $Z=\{Z(t),t\in[0,T]\}$ taking values in a finite space $\mathcal{Z}=\{1,\dots,N\}$, where $N$ is the number of states. We assume that $Z$ is a regular continuous time Markov process and that the starting state $Z(0)\in\mathcal{Z}$ is deterministic. In addition, we assume that $Z$ and $S$ are independent and that $\mathcal{A}=\mathcal{F}_T$ where $\mathcal{F}=\{\mathcal{F}_t\}_{t\in [0,T]}$ is the filtration generated by $Z$ an $S$ satisfying the usual conditions. Whenever we wish to emphasize the marginal filtration, we will denote $\mathcal{F}^Z$ or $\mathcal{F}^S$, respectively.

The transition probabilities of $Z$ are defined for $0\leq s\leq t\leq T$ by 
\begin{align*}
    p_{ij}(s,t)=\PP[Z(t)=j | Z(s)=i],\quad  i,j\in\mathcal{Z},
\end{align*}
and, since $Z$ is a regular Markov process, the transition rates
\begin{align*}
    \mu_{ij}(t) &= \lim_{\substack{h\to 0 \\ h>0}}\frac{p_{ij}(t,t+h)}{h},\quad  i,j\in\mathcal{Z},\quad  i\neq j, \\
    \mu_i(t) &= \lim_{\substack{h\to 0 \\ h>0}}\frac{1-p_{ii}(t,t+h)}{h},\quad  i\in\mathcal{Z},
\end{align*}
exist for every $t\in[0,T]$, are finite and continuous. We define the intrinsic $\mathcal{F}$-adapted processes $\{I_i\}_{i\in \mathcal{Z}}$ and $\{N_{ij}\}_{i,j\in\mathcal{Z}, i\neq j}$ for every $t\in [0,T]$ by
\begin{align*}
I_i(t) &\triangleq \mathbb{I}_{\{Z(t)=i\}},\quad i\in \mathcal{Z},\\
N_{ij}(t) &\triangleq \#\{s\in [0,t]: Z(s-)=i, Z(s)=j\}, \quad i,j\in\mathcal{Z},\quad  i\neq j,
\end{align*}
where $\#$ denotes the counting measure. The random variable $I_i(t)$ tells us whether the insured is in state $i$ at time $t$ or not and the random variable $N_{ij}(t)$ tells us the number of transitions from $i$ to $j$ by time $t$.

The contractual payments between the two parties, insurer and insured, will be modelled by a stochastic process $C$ whose dynamics are given by
\begin{align}\label{eq:cash}
dC(s)= \sum_{i\in\mathcal{Z}} I_i(s) dc_i(s) +\sum_{\substack{i,j\in \mathcal{Z}\\ j\neq i}} c_{ij}(s) dN_{ij}(s),\quad C(0)\in \R,\quad  s\in [0,T],
\end{align}
where $c_i$ are $\mathcal{F}^S$-adapted càdlàg processes with bounded variation and $c_{ij}$ are $\mathcal{F}^S$-predictable processes. Note that $C$ is an $\mathcal{F}$-adapted semimartingale. We assume that $dC(t)$ is null for $t>T$, meaning that the process $C$ stagnates at time $T$.

Typically, $c_i$ model accumulated payments for sojourns in state $i$ and $c_{ij}$ are lump sum benefits for transitions between $i$ and $j$. Moreover, we use the actuarial convention that premiums take negative sign and benefits positive sign.

We define the \emph{present value} of the (stochastic) cash flow $C$ by
$$V(t,C)=\frac{1}{v(t)}\int_{[0,\infty)} v(s)dC(s),\quad t\geq 0.$$
The integral is taken in the Riemann-Stieltjes sense. Note that the integral is well-defined since $C$ is of bounded variation and $v$ is continuous. Moreover, the integration region is actually $[0,T]$ since $dC(s)=0$ for $s>T$.

Furthermore, we introduce the \emph{retrospective} and \emph{prospective} values of $C$ as
\begin{align*}
\cev{V}(t,C)&=\frac{1}{v(t)}\int_{(0,t]} v(s)dC(s),\quad t\geq 0,\\
\vec{V}(t,C)&=\frac{1}{v(t)}\int_{(t,\infty)} v(s)dC(s),\quad t\geq 0.
\end{align*}

It is a trivial fact that
$$V(t,C)=\cev{V}(t,C)+\vec{V}(t,C)$$
for all $t\in [0,T]$ and that
$$V(0,C)=\Delta C(0)+\vec{V}(0,C).$$

Note that if $C$ is $\PP$-a.s. almost everywhere differentiable with finite number of jumps one has the following formula for the present value,
$$V(t,C)=\frac{1}{v(t)} \int_0^\infty v(s)C'(s)ds+\frac{1}{v(t)} \sum_{0\leq s<\infty}v(s)\Delta C(s),$$
which is useful for computational reasons.

We might drop the dependence on $C$ and simply write $V(t)$, $\cev{V}(t)$ and $\vec{V}(t)$ when convenient.

The process $V$ describes the evolution of the value at each time of the entire cash flow $C$, while the process $\cev{V}$ describes the evolution of the \emph{realized} cash flow and $\vec{V}$ describes the evolution of the \emph{future} cash flow. For instance, at a given time $t$, $\cev{V}(t)$ is the $t$-value of the contractual payments that have taken place until $t$ (including $t$) and $\vec{V}(t)$ is the $t$-value of the promised contractual payments from $t$, not including $t$, to the maturity of the contract.

Given the information $\mathcal{F}_t$ at time $t$, the actuary wishes to evaluate the risk of the remaining payments described by the cash flow $C$, whose value is given by $\vec{V}(t,C)$. A typical choice is to compute the so-called net reserve or expected prospective value which is given by
$$\vec{V}_{\mathcal{F}_t}(t,C)\triangleq \mathbb{E}[\vec{V}(t,C)|\mathcal{F}_t], \quad t\in [0,T].$$
Using that $Z$ enjoys the Markov property and that $Z$ and $S$ are independent, we have the following convenient result
\begin{align*}
    \vec{V}_{\mathcal{F}_t}(t,C)=\EE[\vec{V}(t,C)|\mathcal{F}_t]=\EE[\vec{V}(t,C)|\sigma(Z(t))\vee\mathcal{F}_t^S]=H(t,Z(t),S_t),
\end{align*}
where $H$ is a deterministic function. Then, $\vec{V}_{\mathcal{F}_t}(t,C)= \sum_{i\in \mathcal{Z}} I_i(t) H(t,i,S_t)$ and we will write $\vec{V}_i(t,S_t)\triangleq H(t,i,S_t)$ omitting the dependence on $C$. It is noteworthy that the net reserve $\vec{V}_i(t,S_t)$ is a non-anticipative functional of the risky asset $S$. The actuary may pick any possible non-anticipative performance of the asset $S$, say $\omega$, and the value $\vec{V}_i(t,\omega_t)$ would account for the net reserve assuming that the insured is in state $i$ at time $t$ and the market has had the outcome $\omega_t$ by time $t$. Remark that $\vec{V}$, given a state, is a function of time and path, hence an infinite-dimensional argument is needed, in contrast to the classical case. 

\section{A path dependent Thiele's equation}\label{sec:thiele}

In what follows, we want to derive a path dependent differential equation for the net reserve $\vec{V}_{\mathcal{F}_t}$, $t\in [0,T]$. We will study the case when the payments described by $c_i$ and $c_{ij}$ depend on the performance of a fund or risky asset modelled by $S$ in a functional way. 

In the ensuing paragraphs, we will assume the following structure on the processes $c_i$ and $c_{ij}$ appearing in the dynamics of the cash flow $C$ in \eqref{eq:cash}.

\begin{ass}
We assume the following properties for the processes $c_i$ and $c_{ij}$.

\begin{enumerate}\label{hyp1}
\item(Functional assumptions): The càdlàg processes $c_i$ are assumed to be a.e. differentiable with a.e. derivative $c_i'$ and possess a finite number of jumps of finite size at fixed deterministic times $t_0<\cdots <t_n$, where $t_0=0$ and $t_n=T$, conventionally.

Moreover, we assume that there are non-anticipative functionals $f_i$, $g_i$ and $h_{ij}$ such that the jump sizes $\Delta c_i$, the sojourn payments $c_i'$ and the transition payments $c_{ij}$ take the following form
$$\Delta c_i(t_k)= f_i(t_k,S_{t_k}),\quad c_i'(s)=g_i(s,S_s), \quad c_{ij}(s)= h_{ij}(s,S_s)$$
for every $k=0,\dots,n$, $i,j\in \mathcal{Z}$, $i\neq j$ and $s\in [0,T]$.

\item(Finite expectation): The non-anticipative functionals $f_i$, $g_i$ and $h_{ij}$ satisfy
\begin{align*}
    \EE^\QQ\left[|f_i(t,S_t)|\right]<\infty, \quad\EE^\QQ\left[|g_i(t,S_t)|\right]<\infty, \quad\EE^\QQ\left[|h_{ij}(t,S_t)|\right]<\infty,
\end{align*}
for all $t\in[0,T]$, $i,j\in\mathcal{Z}$, $i\neq j$ and $s\in [0,T]$. In addition, 
\begin{align*}
    \int_t^T v(s)p_{ij}(t,s)\EE^\QQ[g_j(s,S_s)|\mathcal{F}_t]ds<\infty, \quad \text{a.s.}\\
    \int_t^T v(s)p_{ij}(t,s)\mu_{jk}(s)\EE^\QQ[h_{jk}(s,S_s)|\mathcal{F}_t]ds<\infty, \quad \text{a.s.},
\end{align*}
for all $t\in[0,T]$, $i,j,k\in\mathcal{Z}$, $j\neq k$.
\end{enumerate}
\end{ass}

\begin{rem}
Hypothesis \ref{hyp1} suggests a representation of the cash flow $C$ that preserves adaptedness to the filtration $\mathcal{F}$ and that allows for functional dependence on the financial asset. For instance, hedging with Asian-type or look-back options is allowed. The functional dependence on the path $S_t$ is general, as long as $f_i$, $g_i$ and $h_{ij}$ are non-anticipative functionals. This also includes the classical unit-linked policies if one takes as $f_i$, $g_i$ and $h_{ij}$ the evaluations of $S$ at $t$.
\end{rem}

Under Hypothesis \ref{hyp1}, the cash flow has the following differential representation:

$$dC(s) = \sum_{i\in \mathcal{Z}} I_i(s) \left(\sum_{k=0}^n f_i(t_k,S_{t_k})\delta_{t_k}(ds)+g_i(s,S_s)ds\right) + \sum_{\substack{i,j\in \mathcal{Z}\\ j\neq i}} h_{ij}(s,S_s)dN_{ij}(s),$$
where $\delta_{t_k}$ is the Dirac measure at $t_k$.


As a consequence of the assumptions imposed on $C$, the prospective value takes the following form,
\begin{align*}
\vec{V}(t) =& \, \frac{1}{v(t)}\sum_{i\in \mathcal{Z}} \sum_{k=0}^n v(t_k)I_i(t_k) f_i(t_k,S_{t_k})\mathbb{I}_{\{t < t_k\}}+\frac{1}{v(t)} \sum_{i\in \mathcal{Z}} \int_t^T v(s) I_i(s)g_i(s,S_s)ds \\
&+ \frac{1}{v(t)}\sum_{\substack{i,j\in \mathcal{Z}\\ j\neq i}} \int_t^T v(s) h_{ij}(s,S_s)dN_{ij}(s). 
\end{align*}

We can see the quantities $f_i(t_k,S_{t_k})$, $g_i(s,S_s)$ and $h_{ij}(s,S_s)$ as contingent claims on $S$ at their respective maturity times $t_k\in[0,T]$, $k=0,\dots,n$ and $s\in[0,T]$. By the arguments in \cite{Steffensen00}, the expected prospective value is given by
\begin{align}
    \vec{V}_i(t)=& \,\frac{1}{v(t)}\sum_{j\in\mathcal{Z}}\sum_{k=0}^nv(t_k)p_{ij}(t,t_k)\EE^\QQ[f_j(t_k,S_{t_k})|\mathcal{F}_t]\mathbb{I}_{\{t < t_k\}} \nonumber \\
    &+ \frac{1}{v(t)}\sum_{j\in\mathcal{Z}}\int_t^T v(s)p_{ij}(t,s)\EE^\QQ[g_j(s,S_s)|\mathcal{F}_t]ds \nonumber \\ 
    &+\frac{1}{v(t)}\sum_{\substack{j,k\in \mathcal{Z}\\ k\neq j}}\int_t^T v(s)p_{ij}(t,s)\mu_{jk}(s)\EE^\QQ[h_{jk}(s,S_s)|\mathcal{F}_t]ds. \label{eq:res1}
\end{align}

\begin{lemm}\label{lem:ito}
    Let $s\in[0,T]$ and $\varphi$ be a non-anticipative functional such that $\EE^\QQ\left[|\varphi(s,S_s)|\right]<\infty$. Then, there exists a non-anticipative functional $U^\varphi_s$ such that 
    \begin{align}\label{eq:funct}
        \frac{v(s)}{v(t)}\EE^\QQ\left[\varphi(s,S_s)|\mathcal{F}_t\right]=U^\varphi_s(t,S_t),\quad t\in[0,s].
    \end{align}
In addition, assume that $U^\varphi_s\in\mathbb{C}_b^{1,2}(\mathcal{W}_s)$ and $\mathcal{D}U^\varphi_s\in\mathbb{C}_l^{0,0}(\mathcal{W}_s)$. Then, $U^\varphi_s$  satisfies the following path dependent PDE
    \begin{align}\label{eq:pdu}
        \mathcal{D}{U}_s^\varphi(t,\omega_t)+\omega(t)r(t)\nabla_\omega U_s^\varphi(t,\omega_t)+\frac{1}{2}\sigma(t,\omega_t)^2\nabla_w^2U_s^\varphi(t,\omega_t)=r(t)U_s^\varphi(t,\omega_t)
    \end{align}
    on the topological support of $S$ in $(C^0([0,s],\RR), \|.\|_\infty)$ and final condition $U_s^\varphi(s,\omega_s)=\varphi(s,\omega_s)$.
\end{lemm}
\begin{proof}
First, observe that by independence of $Z$ and $S$ the following holds $\EE^\QQ\left[\varphi(s,S_s)|\mathcal{F}_t\right]=\EE^\QQ\left[\varphi(s,S_s)|\mathcal{F}_t^S\right]$. Then, there exists a non-anticipative functional $\widetilde{U}^\varphi_s$ such that $\EE^\QQ[\varphi(s,S_s)|\mathcal{F}_t]=\widetilde{U}^\varphi_s(t,S_t)$ for $t\in[0,s]$. One can define $U^\varphi_s(t,S_t)=\frac{v(s)}{v(t)}\widetilde{U}^\varphi_s(t,S_t)$ and \eqref{eq:funct} holds for $t\in[0,s]$.

Assume that $U^\varphi_s\in\mathbb{C}_b^{1,2}(\mathcal{W}_s)$ and $\mathcal{D}U^\varphi_s\in\mathbb{C}_l^{0,0}(\mathcal{W}_s)$. Since $r$ is assumed to be càdlàg, then $\widetilde{U}^\varphi_s\in\mathbb{C}_b^{1,2}(\mathcal{W}_s)$ and $\mathcal{D}\widetilde{U}^\varphi_s\in\mathbb{C}_l^{0,0}(\mathcal{W}_s)$. Since $t\to\widetilde{U}^\varphi_s(t,S_t)$ is a $\QQ$-martingale for $t\in[0,s]$, by Lemma \ref{lem:locmar} we have that
\begin{align}\label{eq:pdtildeu}
    \mathcal{D}\widetilde{U}^\varphi_s(t,\omega_t)+\omega(t)r(t)\nabla_\omega\widetilde{U}^\varphi_s(t,\omega_t)+\frac{1}{2}\sigma(t,\omega_t)^2\nabla^2_\omega\widetilde{U}^\varphi_s(t,\omega_t)=0. 
\end{align}
on the topological support of $S$ in $(C^0([0,s],\RR), \|.\|_\infty)$. One can check that
\begin{align*}
    \mathcal{D}U_s^\varphi(t,\omega_t) &=\frac{v(s)}{v(t)}\mathcal{D}\widetilde{U}^\varphi_s(t,\omega_t)+r(t)\frac{v(s)}{v(t)}\widetilde{U}^\varphi_s(t,\omega_t), \\
    \nabla_\omega U^\varphi_s(t,\omega_t) &= \frac{v(s)}{v(t)}\nabla_\omega\widetilde{U}^\varphi_s(t,\omega_t), \\
    \nabla_\omega^2U^\varphi_s(t,\omega_t) &=\frac{v(s)}{v(t)}\nabla_\omega^2\widetilde{U}^\varphi_s(t,\omega_t).
\end{align*}
Replacing the previous equalities in \eqref{eq:pdtildeu} we obtain the path dependent PDE \eqref{eq:pdu}.
\end{proof}

Observe that the expression of the expected prospective reserve given in \eqref{eq:res1} can be written in a more compact way using the functionals introduced in Lemma \ref{lem:ito}:
\begin{align}\label{eq:res2}
    \vec{V}_i(t)=& \,\sum_{j\in\mathcal{Z}}\sum_{k=0}^np_{ij}(t,t_k)U_{t_k}^{f_j}(t,S_t)\mathbb{I}_{\{t < t_k\}} +\sum_{j\in\mathcal{Z}}\int_t^T p_{ij}(t,s)U_s^{g_j}(t,S_t)ds  \notag\\ 
    &+\sum_{\substack{j,k\in \mathcal{Z}\\ k\neq j}}\int_t^T p_{ij}(t,s)\mu_{jk}(s)U_s^{h_{jk}}(t,S_t)ds.
\end{align}
Note that in the previous expression, $\vec{V}_i(t)$ depends on $S_t$, not just on $S(t)$.

\begin{thm}[A path dependent Thiele's partial differential equation]\label{thm:thiele} Let $i\in\mathcal{Z}$. There exists a non-anticipative functional $V_i$ such that $\vec{V}_i(t)=V_i(t,S_t)$ for $t\in[0,T]$. In addition, assume that for all $t\in[0,T]$, 
\begin{enumerate}
    \item $U_{t_l}^{f_j},U_s^{g_j},U_s^{h_{jk}}\in\mathbb{C}_b^{1,2}(\mathcal{W}_T)$ for $j,k\in\mathcal{Z}, j\neq k$, $l=0,...,n$ and all $s\in[0,T]$.
        \item $\mathcal{D}U_{t_l}^{f_j},\mathcal{D}U_s^{g_j},\mathcal{D}U_s^{h_{jk}}\in\mathbb{C}_l^{0,0}(\mathcal{W}_T)$ for $j,k\in\mathcal{Z}, j\neq k$, $l=0,...,n$ and all $s\in[0,T]$.
    \item There exists  non-anticipative functional $H_{s}$ such that 
    \begin{align*}
        \max\{|U_s^{g_j}(u,\omega_t)|,|U_s^{h_{jk}}(u,\omega_t)|,|\mathcal{D}U_s^{g_j}(u,\omega_t)|,|\mathcal{D}U_s^{h_{jk}}(u,\omega_t)|\}\leq H_{s}(t,\omega_t)
    \end{align*}
     on the topological support of $S$ in $(C^0([0,s],\RR), \|.\|_\infty)$, for $j,k\in\mathcal{Z}, j\neq k$, $u,s\in[t,T]$, $u\leq s$ and $\int_t^TH_{s}(t,\omega_t)ds<\infty$.
    \item There exists a non-anticipative functional $J_{s}$ such that
\begin{align*}
     \max_{l=1,2}\{|\nabla_\omega^l U_s^{g_j}(t,\omega_t+u\mathbb{I}_{[t,T]})|,|\nabla_\omega^l U_s^{h_{jk}}(t,\omega_t+u\mathbb{I}_{[t,T]})|\}\leq J_{s}(t,\omega_t)
\end{align*}
    on the topological support of $S$ in $(C^0([0,s],\RR), \|.\|_\infty)$, for $j,k\in\mathcal{Z}, j\neq k$, $u$ in a neighborhood of $0$, $s\in[t,T]$ and $\int_t^TJ_{s}(t,\omega_t)ds<\infty$.
\end{enumerate}
Then, $V_i$ is a solution of the following path dependent PDE
\begin{align*}
\begin{split}
    \mathcal{D}V_i(t,\omega_t)& \,= r(t)V_i(t,\omega_t)-g_i(t,\omega_t) \\
    &-\sum_{j\neq i}\mu_{ij}(t)\left(h_{ij}(t,\omega_t)+V_j(t,\omega_t)-V_i(t,\omega_t)\right)-\mathcal{L}_\omega V_i(t,\omega_t),
    \end{split}
\end{align*}
on the topological support of $S$ in $(C^0([0,T],\RR), \|.\|_\infty)$, with final condition $V_i(T,\omega_T)=0$ and $\mathcal{L}_\omega$ is the path dependent differential operator defined as
\begin{align*}
    \mathcal{L}_\omega F(t,\omega_t)&=\omega(t)r(t)\nabla_\omega F(t,\omega_t) +\frac{1}{2}\sigma(t,\omega_t)^2\nabla_\omega^2F(t,\omega_t) \\
    &=\omega(t)r(t)\nabla_\omega F(t,\omega_t) +\frac{1}{2}\omega(t)^2\widetilde{\sigma}(t,\omega_t)^2\nabla_\omega^2F(t,\omega_t) 
\end{align*}
    for $F\in\mathbb{C}_b^{1,2}(\mathcal{W}_T)$.
\end{thm}

\begin{proof}
    First of all and to lighten the proof, we remark that all the equalities in the proof involving non-anticipative functionals hold on the topological support of $S$ in $(C^0([0,T],\RR),\|.\|_\infty)$. Note that the functional $V_i$ is already given in \eqref{eq:res2}, that is, 
    \begin{align*}
        V_i(t,\omega_t)=& \,\sum_{j\in\mathcal{Z}}\sum_{k=0}^np_{ij}(t,t_k)U_{t_k}^{f_j}(t,\omega_t)\mathbb{I}_{\{t < t_k\}} +\sum_{j\in\mathcal{Z}}\int_t^T p_{ij}(t,s)U_s^{g_j}(t,\omega_t)ds  \notag\\ 
    &+\sum_{\substack{j,k\in \mathcal{Z}\\ k\neq j}}\int_t^T p_{ij}(t,s)\mu_{jk}(s)U_s^{h_{jk}}(t,\omega_t)ds.
    \end{align*}
   Assume now that $U_{t_l}^{f_j},U_s^{g_j},U_s^{h_{jk}}\in\mathbb{C}_b^{1,2}(\mathcal{W}_T)$ for $j,k\in\mathcal{Z}, j\neq k$, $l=0,...,n$ and all $s\in[0,T]$. Furthermore, assume that $\mathcal{D}U_{t_l}^{f_j},\mathcal{D}U_s^{g_j},\mathcal{D}U_s^{h_{jk}}\in\mathbb{C}_l^{0,0}(\mathcal{W}_T)$ for $j,k\in\mathcal{Z}, j\neq k$, $l=0,...,n$ and all $s\in[0,T]$. We can write
    \begin{align*}
        V_i(t,\omega_t)=G_i(t,\omega_t)+\int_t^TF_{i,s}(t,\omega_t)ds,
    \end{align*}
    where
    \begin{align}
        G_i(t,\omega_t) &=\sum_{j\in\mathcal{Z}}\sum_{k=0}^np_{ij}(t,t_k)U_{t_k}^{f_j}(t,\omega_t)\mathbb{I}_{\{t < t_k\}},  \nonumber \\
        F_{i,s}(t,\omega_t) &=\sum_{j\in\mathcal{Z}}p_{ij}(t,s)\left(U_s^{g_j}(t,\omega_t)+\sum_{k\neq i}\mu_{jk}(s)U_s^{h_{jk}}(t,\omega_t)\right) \nonumber \\
        &=\sum_{j\in\mathcal{Z}}p_{ij}(t,s)U_s^{\psi_j}(t,\omega_t), \label{eq:defF}
    \end{align}
    where $\psi_j(s,\omega_s)=g_j(s,\omega_s)+\sum_{k\neq j}\mu_{jk}(s)h_{jk}(s,\omega_s)$. Since $Z$ is a regular Markov process and by the regularity assumptions on the functionals $U$, we have that $U_s^{\psi_j}, F_{i,s}\in\mathbb{C}^{1,2}_b(\mathcal{W}_T)$ for all $s\in[0,T]$, $\mathcal{D}U_s^{\psi_j}\in\mathbb{C}_l^{0,0}(\mathcal{W}_T)$ and $G_i\in\mathbb{C}_{b}^{1,2}(\mathcal{W}_T)$.  Recall Kolmogorov's backward equation
\begin{align*}
    \partial_t p_{ij}(t,s)=\sum_{k\neq i}\mu_{ik}(t)\left(p_{ij}(t,s)-p_{kj}(t,s)\right). 
\end{align*}
Then, applying the product rule and Kolmogorov's backward equation
\begin{align}
    \mathcal{D}F_{i,s}(t,\omega_t)& \,=\sum_{j\in\mathcal{Z}}\partial_tp_{ij}(t,s)U_{s}^{\psi_j}(t,\omega_t)+\sum_{j\in\mathcal{Z}}p_{ij}(t,s)\mathcal{D}U_s^{\psi_j}(t,\omega_t) \nonumber \\ 
    & \,=\sum_{j\in\mathcal{Z}}\sum_{k\neq i}\mu_{ik}(t)\left(p_{ij}(t,s)-p_{kj}(t,s)\right)U_s^{\psi_j}(t,\omega_t)\nonumber \\ 
    &+\sum_{j\in\mathcal{Z}}p_{ij}(t,s)\mathcal{D}U_s^{\psi_j}(t,\omega_t)\nonumber \\
    & \,=\sum_{k\neq i}\mu_{ik}(t)\left(F_{i,s}(t,\omega_t)-F_{k,s}(t,\omega_t)\right)\nonumber \\ 
    &+\sum_{j\in\mathcal{Z}}p_{ij}(t,s)\mathcal{D}U_s^{\psi_j}(t,\omega_t). \label{eq:DFexpr}
\end{align}
Moreover, by linearity of $\mathcal{L}_\omega$,
\begin{align}\label{eq:DVFexpr}
    \mathcal{L}_\omega F_{i,s}(t,\omega_t)=\sum_{j\in\mathcal{Z}}p_{ij}(t,s)\mathcal{L}_\omega U_s^{\psi_j}(t,\omega_t).
\end{align}
Note that by Lemma \ref{lem:ito} the following is satisfied
\begin{align*}
    \mathcal{D}U_s^{\psi_j}(t,\omega_t)+\mathcal{L}_\omega U_s^{\psi_j}(t,\omega_t)=r(t)U_s^{\psi_j}(t,\omega_t).
\end{align*}
Therefore, 
\begin{align}\label{eq:F1}
    \mathcal{D}F_{i,s}(t,\omega_t)+\mathcal{L}_\omega F_{i,s}(t,\omega_t) = \sum_{k\neq i}\mu_{ik}(t)\left(F_{i,s}(t,\omega_t)-F_{k,s}(t,\omega_t)\right)+r(t)F_{i,s}(t,\omega_t).
\end{align}
Recall that 
 \begin{align*}
        V_i(t,\omega_t)=G_i(t,\omega_t)+\int_t^TF_{i,s}(t,\omega_t)ds.
    \end{align*}
Now, 
\begin{align*}
    \mathcal{D}V_i(t,\omega_t)=\mathcal{D}G_i(t,\omega_t)+\mathcal{D}\int_t^TF_{i,s}(t,\omega_t)ds.
\end{align*}
We now justify that 
\begin{align*}
    \mathcal{D}\int_t^TF_{i,s}(t,\omega_t)ds=\int_t^T\mathcal{D}F_{i,s}(t,\omega_t)ds-F_{i,t}(t,\omega_t).
\end{align*}
Define the functional $K_{i}(u,\omega_t):=\int_u^TF_{i,s}(u,\omega_t)ds$ for $u\in[t,T]$. Then, 
\begin{align*}
      \mathcal{D}\int_t^TF_{i,s}(t,\omega_t)ds= \lim_{\substack{h\to 0 \\ h>0}} \frac{1}{h}\left(K_{i}(t+h,\omega_t)-K_{i}(t,\omega_t)\right).  
\end{align*}
One can check that $\partial_uF_{i,s}(u,\omega_t)=\mathcal{D}F_{i,s}(u,\omega_t)$. Indeed, 
\begin{align*}
    \partial_u F_{i,s}(u,\omega_t)&=\lim_{h\to 0}\frac{1}{h}\left(F_{i,s}(u+h,\omega_t)-F_{i,s}(u,\omega_t)\right) \\
    & = \lim_{h\to 0}\frac{1}{h}\left(F_{i,s}(u+h,(\omega_t)_u)-F_{i,s}(u,(\omega_t)_u)\right) \\
    &= \mathcal{D}F_{i,s}(u,(\omega_t)_u) =  \mathcal{D}F_{i,s}(u,\omega_t).
\end{align*}
Then, by the assumptions in (2), the expression of $F_{i,s}$ in \eqref{eq:defF}, the one of $\mathcal{D}F_{i,s}$ in \eqref{eq:DFexpr} and using that $Z$ is a regular Markov process, we have that $|\partial_uF_{i,s}(u,\omega_t)|=|\mathcal{D}F_{i,s}(u,\omega_t)|\leq L_{s}(t,\omega_t)$ for $u,s\in[t,T]$, $u\leq s$ where $L_{s}$ is satisfies $\int_t^TL_{s}(t,\omega_t)ds<\infty$. Therefore, we can differentiate under the integral sign to get
\begin{align*}
      \mathcal{D}\int_t^TF_{i,s}(t,\omega_t)ds=\int_t^T\mathcal{D}F_{i,s}(t,\omega_t)ds-F_{i,t}(t,\omega_t).
\end{align*}
Therefore, 
\begin{align*}
\mathcal{D}V_i(t,\omega_t)=\mathcal{D}G_i(t,\omega_t)+\int_t^T\mathcal{D}F_{i,s}(t,\omega_t)-F_{i,t}(t,\omega_t).
\end{align*}
Observe, 
\begin{align*}
    F_{i,t}(t,\omega_t)=\sum_{j\in\mathcal{Z}}p_{ij}(t,t)U_t^{\psi_j}(t,\omega_t)=U_t^{\psi_i}(t,\omega_t)=g_i(t,\omega_t)+\sum_{k\neq i}\mu_{ik}(t)h_{ik}(t,\omega_t).
\end{align*}
Altogether
\begin{align}\label{eq:intDF}
    \int_t^T\mathcal{D}F_{i,s}(t,\omega_t)ds = \mathcal{D}V_i(t,\omega_t)-\mathcal{D}G_i(t,\omega_t)+g_i(t,\omega_t)+\sum_{k\neq i}\mu_{ik}(t)h_{ik}(t,\omega_t).
\end{align}
Integrating \eqref{eq:F1} with respect to $s$ on the region $[t,T]$
\begin{align}
    \int_t^T\mathcal{D}F_{i,s}(t,\omega_t)ds+\int_t^T\mathcal{L}_\omega F_{i,s}(t,\omega_t)ds& = \nonumber \\ \sum_{k\neq i}\mu_{ik}(t)\int_t^T\left(F_{i,s}(t,\omega_t)-F_{k,s}(t,\omega_t)\right)ds+r(t)\int_t^TF_{i,s}(t,\omega_t)ds. \label{eq:intDF2}
\end{align}
We know justify that the following holds $\int_t^T\mathcal{L}_\omega F_{i,s}(t,\omega_t)ds=\mathcal{L}_\omega\int_t^T F_{i,s}(t,\omega_t)ds$. It is enough to prove that
\begin{align*}
    \int_t^T\nabla^l_\omega F_{i,s}(t,\omega_t)ds = \nabla^l_\omega\int_t^TF_{i,s}(t,\omega_t)ds
\end{align*}
for $l=1,2$.
We first prove that $\nabla_\omega\int_t^TF_{i,s}(t,\omega_t)ds=\int_t^T\nabla_\omega F_{i,s}(t,\omega_t)ds$. Define the functional $M_{i,t}(u,\omega_t)=\int_t^TF_{i,s}(t,\omega_t+u\mathbb{I}_{[t,T]})ds$ for $u$ in a neighborhood of $0$. Then, 
\begin{align*}
    \nabla_\omega\int_t^TF_{i,s}(t,\omega_t)ds = \lim_{h\to0}\frac{1}{h}\left(M_{i,t}(h,\omega_t)-M_{i,t}(0,\omega_t)\right).
\end{align*}
By the assumptions in (3), the expression of $\nabla_\omega F_{i,s}(t,\omega_t)$ in \eqref{eq:DVFexpr} and using that $Z$ is a regular Markov process, one can check that $|\partial_uF_{i,s}(t,\omega_t+u\mathbb{I}_{[t,T]})|=|\nabla_\omega F_{i,s}(t,\omega_t+u\mathbb{I}_{[t,T]})|\leq N_{s}(t,\omega_t)$ for $u$ in a neighborhood of $0$ and $s\in[t,T]$ where $N_{s}$ satisfies $\int_t^TN_{s}(t,\omega_t)ds<\infty$. Therefore, we can differentiate under the integral sign to get 
\begin{align*}
     \nabla_\omega\int_t^TF_{i,s}(t,\omega_t)ds = \int_t^T\nabla_\omega F_{i,s}(t,\omega_t)ds.
\end{align*}
Similarly, one can prove that $\nabla^2_\omega\int_t^TF_{i,s}(t,\omega_t)ds=\int_t^T\nabla^2_\omega F_{i,s}(t,\omega_t)ds$. Replacing \eqref{eq:intDF} in \eqref{eq:intDF2} and interchanging the operator $\mathcal{L}_\omega$ with the integral we get
\begin{align*}
    \mathcal{D}V_i(t,\omega_t)-\mathcal{D}G_i(t,\omega_t)+g_i(t,\omega_t)+\sum_{k\neq i}\mu_{ik}(t)h_{ik}(t,\omega_t)+\mathcal{L}_\omega\left(V_{i,s}(t,\omega_t)-G_{i,s}(t,\omega_t)\right)& = \\
    \sum_{k\neq i}\mu_{ik}(t)\left(V_i(t,\omega_t)-V_k(t,\omega_t)+G_k(t,\omega_t)-G_i(t,\omega_t)\right)+r(t)\left(V_i(t,\omega_t)-G_i(t,\omega_t)\right).
\end{align*}
One can check that the terms involving $G$ in the previous equation cancel out. Indeed, 
\begin{align*}
    \mathcal{D}G_i(t,\omega_t)& \,=\sum_{j\in\mathcal{Z}}\sum_{k=0}^n\partial_tp_{ij}(t,t_k)U_{t_k}^{f_j}(t,\omega_t)\mathbb{I}_{\{t < t_k\}}+\sum_{j\in\mathcal{Z}}\sum_{k=0}^np_{ij}(t,t_k)\mathcal{D}U_{t_k}^{f_j}(t,\omega_t)\mathbb{I}_{\{t < t_k\}} \\
    &=\sum_{j\in\mathcal{Z}}\sum_{k=0}^n\sum_{l\neq i}\mu_{il}(t)\left(p_{ij}(t,t_k)-p_{lj}(t,t_k)\right)U_{t_k}^{f_j}(t,\omega_t)\mathbb{I}_{\{t < t_k\}} \\
    & \,+\sum_{j\in\mathcal{Z}}\sum_{k=0}^np_{ij}(t,t_k)\mathcal{D}U_{t_k}^{f_j}(t,\omega_t)\mathbb{I}_{\{t < t_k\}} \\
    &=\sum_{l\neq i}\mu_{il}(t)\sum_{j\in\mathcal{Z}}\sum_{k=0}^n\left(p_{ij}(t,t_k)-p_{lj}(t,t_k)\right)U_{t_k}^{f_j}(t,\omega_t)\mathbb{I}_{\{t < t_k\}} \\
    & \,+\sum_{j\in\mathcal{Z}}\sum_{k=0}^np_{ij}(t,t_k)\mathcal{D}U_{t_k}^{f_j}(t,\omega_t)\mathbb{I}_{\{t < t_k\}} \\
    &=\sum_{l\neq i}\mu_{il}(t)\left(G_i(t,\omega_t)-G_l(t,\omega_t)\right)+\sum_{j\in\mathcal{Z}}\sum_{k=0}^np_{ij}(t,t_k)\mathcal{D}U_{t_k}^{f_j}(t,\omega_t)\mathbb{I}_{\{t < t_k\}}
\end{align*}
Therefore,
\begin{align*}
    \mathcal{D}G_i(t,\omega_t)+\mathcal{L}_\omega G_i(t,\omega_t) & \,= \sum_{l\neq i}\mu_{il}(t)\left(G_i(t,\omega_t)-G_l(t,\omega_t)\right) \\
    & \,+\sum_{j\in\mathcal{Z}}\sum_{k=0}^np_{ij}(t,t_k)\left(\mathcal{D}U_{t_k}^{f_j}(t,\omega_t)+\mathcal{L}_\omega U_{t_k}^{f_j}(t,\omega_t)\right)\mathbb{I}_{\{t < t_k\}} \\
    &= \sum_{l\neq i}\mu_{il}(t)\left(G_i(t,\omega_t)-G_l(t,\omega_t)\right) \\ 
    & \,+\sum_{j\in\mathcal{Z}}\sum_{k=0}^np_{ij}(t,t_k)r(t)U_{t_k}^{f_j}(t,\omega_t)\mathbb{I}_{\{t < t_k\}} \\ 
        &= \sum_{l\neq i}\mu_{il}(t)\left(G_i(t,\omega_t)-G_l(t,\omega_t)\right) + r(t)G_i(t,\omega_t).
\end{align*}
where we have used that $\mathcal{D}U_{t_k}^{f_j}(t,\omega_t)+\mathcal{L}_\omega U_{t_k}^{f_j}(t,\omega_t)=r(t)U_{t_k}^{f_j}(t,\omega_t)$. Finally, we get
\begin{align*}
    \mathcal{D}V_i(t,\omega_t)& \,= r(t)V_i(t,\omega_t)-g_i(t,\omega_t) \\
    &-\sum_{j\neq i}\mu_{ij}(t)\left(h_{ij}(t,\omega_t)+V_j(t,\omega_t)-V_i(t,\omega_t)\right)-\mathcal{L}_\omega V_i(t,\omega_t).
\end{align*}
\end{proof}

We give a simple payoff example and check that all the assumptions in  Theorem \ref{thm:thiele} are satisfied.

\begin{exam} We consider the payoff 
\begin{align*}
    \varphi(s,S_s)=\frac{1}{s}\int_0^sS(v)dv,
\end{align*}
where $S$ follows the Black Scholes model 
\begin{align*}
    dS(t)=\mu S(t)dt+\sigma S(t)dW(t).
\end{align*}
For the sake of simplicity, we assume a constant interest rate. Then, $S$ has dynamics 
\begin{align*}
    dS(t)=rS(t)dt+\sigma S(t)dW^\QQ(t),
\end{align*}
under the unique equivalent martingale measure $\QQ$ and $W^\QQ$ is a $\QQ$-Brownian motion. 

One can check that the functional $U_s^\varphi$ defined in Lemma \ref{lem:ito} by
\begin{align*}
    U_s^\varphi(t,\omega_t)=\frac{v(s)}{v(t)}\EE^\QQ[\varphi(s,S_s)|\mathcal{F}_t]=e^{-r(s-t)}\EE^\QQ\left[\frac{1}{s}\int_0^sS(v)dv|\mathcal{F}_t\right],
\end{align*}
is given by
\begin{align*}
     U_s^\varphi(t,\omega_t)=\frac{e^{-r(s-t)}}{s}\int_0^t\omega(v)dv+\frac{\omega(t)}{rs}\left(1-e^{-r(s-t)}\right).
\end{align*}
Then, one can verify that $U_s^\varphi\in\mathbb{C}_b^{1,2}(\mathcal{W}_s)$ and $\mathcal{D}U_s^\varphi\in\mathbb{C}_l^{0,0}(\mathcal{W}_s)$ where the horizontal and vertical derivatives are given by
\begin{align*}
    \mathcal{D}U_s^\varphi(t,\omega_t)&=\frac{re^{-r(s-t)}}{s}\int_0^t\omega(v)dv, \\
    \nabla U_s^\varphi(t,\omega_t)&=\frac{1}{rs}\left(1-e^{-r(s-t)}\right), \\
     \nabla^2 U_s^\varphi(t,\omega_t)&=0.
\end{align*}
In order to check assumption (3) in Theorem \ref{thm:thiele}, one can see that for $t\leq u\leq s$, 
\begin{align*}
    |U_s^\varphi(u,\omega_t)|=\left|\frac{e^{-r(s-u)}}{s}\int_0^u\omega_t(v)dv+\frac{\omega(t)}{rs}\left(1-e^{-r(s-u)}\right)\right| \\
    \leq \frac{1}{s}\int_0^s|\omega_t(v)|dv+\frac{|\omega(t)|}{rs}\left(1-e^{-rs}\right) \\
    \leq \sup_{v\in[0,t]}|\omega(v)|\left(1+\frac{1}{rs}\left(1-e^{-rs}\right)\right),
\end{align*}
where $\int_0^T\left(1+\frac{1}{rs}\left(1-e^{-rs}\right)\right)ds<\infty$. One can similarly proceed with the horizontal derivative, the first derivative and the second vertical derivative.

\end{exam}

\end{document}